\documentclass[aps,prl,reprint,10pt,raggedbottom,superscriptaddress,nofootinbib]{revtex4-2}
\usepackage{amsmath,amssymb,amsthm}
\usepackage{graphicx}
\usepackage{microtype}
\usepackage{color}
\definecolor{linkblue}{rgb}{0.10,0.25,0.48}
\usepackage[colorlinks=true,allcolors=linkblue]{hyperref}
\hypersetup{pdftitle={Preserving measurement incompatibility forces noiseless codes},pdfauthor={Yu-Xuan Zhang and Jing-Ling Chen}}
\makeatletter
\let\auto@bib@innerbib\@empty
\makeatother
\newtheorem{theorem}{Theorem}
\newtheorem{corollary}{Corollary}

\DeclareMathOperator{\Tr}{Tr}
\DeclareMathOperator{\ran}{ran}

\newcommand{\id}{\mathbb I}
\newcommand{\C}{\mathbb C}

\newcommand{\N}{\mathcal N}
\newcommand{\ket}[1]{\lvert #1\rangle}
\newcommand{\bra}[1]{\langle #1\rvert}
\newcommand{\proj}[1]{\ket{#1}\bra{#1}}
\newcommand{\norm}[1]{\lVert #1\rVert}

\begin{document}
\lefthyphenmin=2
\righthyphenmin=3
\title{Preserving measurement incompatibility forces noiseless codes}
\author{Yu-Xuan Zhang}
\affiliation{School of Physics, Nankai University, Tianjin 300071, People's Republic of China}
\author{Jing-Ling Chen}
\email{chenjl@nankai.edu.cn}
\affiliation{Theoretical Physics Division, Chern Institute of Mathematics, Nankai University, Tianjin 300071, People's Republic of China}
\date{\today}
\begin{abstract}
We construct $d^2$ binary measurements for each output dimension $d\ge2$ and tune one until they cease to admit a joint measurement. Preserving this boundary forces a channel to carry a noiseless copy of its entire output space, with every superposition transmitted intact and the environment left in the same state for every encoded input. Such a code is known to preserve every incompatible measurement family, and we prove the converse for arbitrary finite dimensions. Furthermore, for every $d\ge2$, there exists a family of binary measurements that is incompatible on the full space but compatible after compression to any proper subspace. We also construct the encoding explicitly and show that, for a sufficiently small shift of the compatibility threshold, the worst-case transmission error is at most quadratic in the shift. On the shared state associated with the channel, this boundary marks the onset of steering and decides whether every incompatible family can demonstrate steering on that one state.
\end{abstract}
\maketitle
\clubpenalty=10000
\widowpenalty=10000
\brokenpenalty=10000

\emph{Introduction.---}
Quantum steering and quantum communication put different demands on the same noisy link. Schr\"odinger~\cite{Schrodinger1935}, following Einstein, Podolsky, and Rosen~\cite{EPR1935}, observed that measurements on one half of an entangled pair prepare different ensembles on the other without changing its average state~\cite{Hughston1993}. For steering, no fixed ensemble of local states can account for those preparations~\cite{Wiseman2007,Jones2007,Cavalcanti2017,Uola2020}. When a link transmits an entire quantum system faithfully~\cite{Shor1995,Steane1996,KnillLaflamme1997}, it retains every measurement choice capable of steering, but does preserving those choices force it to carry the states themselves?

Measurement incompatibility means that no single parent measurement can reproduce the outcome statistics for all those choices by classical processing~\cite{Busch1986,Heinosaari2016,Guhne2023}. A noiseless code means that one encoding transmits every state of the output system intact, including every superposition~\cite{ZanardiRasetti1997,Lidar1998,Knill2000}. A channel that preserves every incompatible family must, we prove, contain such a code. For a $d$-dimensional output, one can decide which channels do by varying just one of $d^2$ suitably chosen binary measurements (Fig.~\ref{fig:boundary-code}). The last setting for which the ideal measurements admit a parent gives the reference value. If a parent after the noise reaches further, there is an interval of settings where the measurements are still incompatible but the noisy pair shows no steering. The width of a narrow interval also controls the worst-case transmission error.

\begin{figure*}[t]
\centering
\includegraphics[trim={68bp 198bp 85bp 158bp},clip,width=0.8\textwidth]{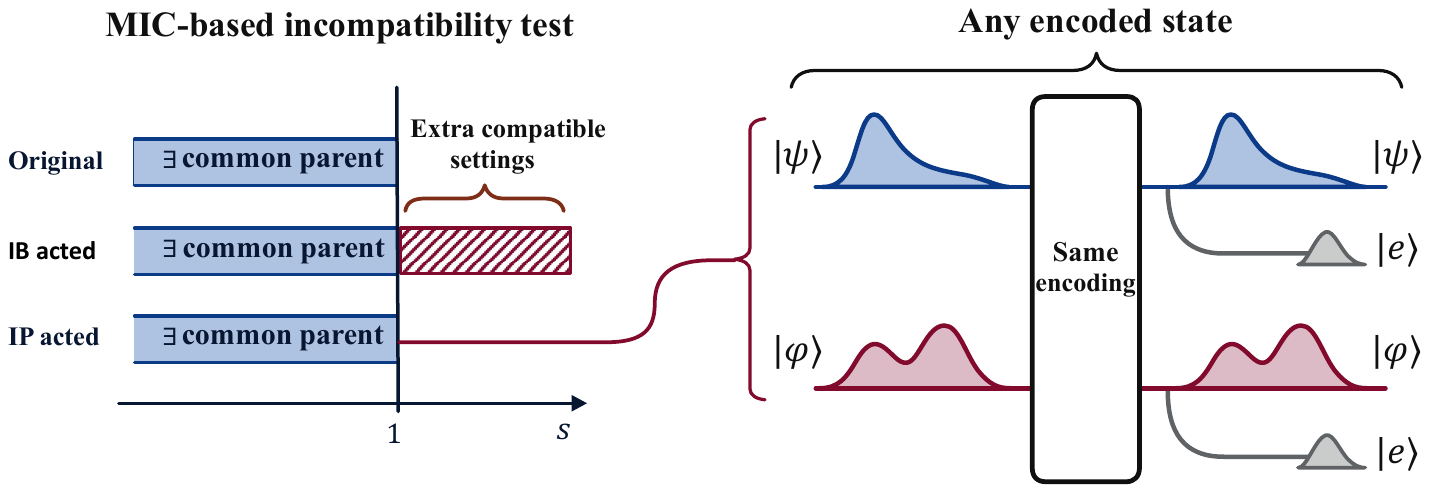}
\caption{At left, $s$ tunes one of $d^2$ binary measurements built from a minimal informationally complete measurement (MIC). The original family admits a common parent up to $s=1$. An incompatibility-breaking (IB) channel extends that range into the hatched interval, while an incompatibility-preserving (IP) channel leaves the boundary at $1$. Preserving this one boundary forces a noiseless code for the entire output space (right). The same encoding transmits $\lvert\psi\rangle$, $\lvert\varphi\rangle$, and every superposition intact, leaving the
environment in $\lvert e\rangle$ for every logical input.}
\label{fig:boundary-code}
\end{figure*}

To relate this threshold to steering, we use the Heisenberg adjoint to pull the output measurement effects back to the input. These effects are the positive operators specifying the outcomes. The adjoint takes a parent measurement to a parent measurement and hence preserves compatibility. Incompatibility can be witnessed by an advantage in a suitable state-discrimination task~\cite{Carmeli2019,Skrzypczyk2019,Uola2019} and, for every incompatible family, by steering on a maximally entangled pair. In that pair the conditional states on one side are proportional to the transposed effects on the other, so a parent and a local hidden-state explanation determine each other~\cite{Quintino2014,Uola2014}. In the state associated with a noisy channel, normalizing the conditional states with their marginal~\cite{Uola2015,Ku2022} gives the pulled-back measurements up to transposition~\cite{Kiukas2017}, so the onset of steering on that one state measures the compatibility threshold.

The effects of quantum noise on measurement incompatibility and steering have long been studied. Heinosaari et al.~\cite{Heinosaari2015} identified channels that destroy all measurement incompatibility and with it the ability to demonstrate steering on their associated states. Steering inequalities~\cite{Reid1989,Cavalcanti2009,Skrzypczyk2014} detect the correlations lost in this limit, a resource~\cite{GallegoAolita2015} for key distribution~\cite{Branciard2012} and subchannel discrimination~\cite{PianiWatrous2015}. Chitambar and Zhang~\cite{ChitambarZhang} also studied channels that preserve every incompatible measurement family. They derived unitarity in equal dimensions from classification results for maps preserving compatibility relations~\cite{Geher2020,Mori2021} and conjectured that for a larger input the adjoint must contain an identity block. Such a block carries the full output space, with an environment state that is the same for every encoded input and can vary for inputs outside the code~\cite{SchumacherNielsen1996,Kretschmann2008}. A channel that discards one input subsystem illustrates this case. Its pullback tensors each output effect with the identity on the discarded subsystem, so every incompatible family survives.

To prove that preserving the compatibility threshold of our binary family forces a noiseless code, we choose the fixed measurements so that their effects uniquely determine the common parent, a minimal informationally complete measurement. The remaining freedom is how to split each parent outcome between the two outcomes of the variable measurement. We increase its parameter $s$ until a coefficient needed for the splitting turns negative, which happens precisely for $s>1$. We show that failure to extend this splitting beyond $s=1$ selects a single environment state that accompanies every superposition in a copy of the output space, giving an explicit encoding that preserves the joint state of the logical system and any reference system.

The same construction also extends the three-dimensional example of Uola et al.~\cite{Uola2021}, in which an incompatible measurement family becomes compatible on every two-dimensional subspace. Measurement simulation and high-dimensional steering give operational notions of the dimension required by measurements~\cite{Ioannou2022,Designolle2021,Jones2023}. We show that for each $d\ge2$ one setting of our binary measurements is incompatible on the full space and compatible on every proper subspace. Projecting out any one direction already suffices. We also quantify the stability of this encoding, relating noise robustness of measurements~\cite{HeinosaariKiukasReitzner2015,Designolle2019} to approximate quantum error correction~\cite{BenyOreshkov2010,Zheng2024}. If the threshold shifts sufficiently little, the worst-case transmission error is bounded by a constant times the square of that shift.

This letter is organised as follows. We first express steering on a fixed shared state in terms of measurements pulled back through a channel. Then we construct the binary measurements and their common parent measurement. Finally, we extract the encoding from the threshold and derive both the subspace and stability results.

\emph{Steering and incompatibility preservation.---}
For a bipartite state, a measurement on one subsystem prepares a conditional state on the other, depending on the outcome. Let $j$ label the measurement and $a$ its outcome, and denote the conditional state weighted by its outcome probability by $\sigma_{a|j}$. These conditional states sum to the reduced state $\rho_B$ of the steered party, independently of the measurement choice. If they admit a local hidden-state model,
\begin{equation}
 \sigma_{a|j}=\sum_\lambda p(a|j,\lambda)\tau_\lambda, \qquad \sum_\lambda\tau_\lambda=\rho_B,
\end{equation}
the family of conditional states is called unsteerable; otherwise, these measurements demonstrate quantum steering. Here $\operatorname{Tr}\tau_\lambda$ is the weight of each hidden state.

On the support of $\rho_B$, we normalize this decomposition by multiplying each operator from the left and right by $\rho_B^{-1/2}$. Denoting the normalized conditional states and hidden states by $E_{a|j}$ and $G_\lambda$, respectively, we obtain
\begin{equation}
 E_{a|j}=\sum_\lambda p(a|j,\lambda)G_\lambda, \qquad \sum_\lambda G_\lambda=I.
\end{equation}
Here $\{G_\lambda\}$ is a measurement, and each measurement $\{E_{a|j}\}_a$ can be obtained from it by classical postprocessing. This is precisely the definition of joint measurability, or compatibility, and $\{G_\lambda\}$ is called a parent measurement. Since this normalization is reversible, the family of conditional states is unsteerable if and only if the normalized measurement family is jointly measurable~\cite{Uola2015}.

Let $\N$ be a quantum channel from a $D$-dimensional system to a $d$-dimensional system, and write $\Lambda=\N^*$ for its Heisenberg dual. For an input state $\rho$ and an output operator $X$,
\begin{equation}
 \operatorname{Tr}[\N(\rho)X]=\operatorname{Tr}[\rho\Lambda(X)].
\end{equation}
Thus measuring the effect $M_{a|j}$ at the output gives the same statistics as measuring $\Lambda(M_{a|j})$ at the input. Since $\Lambda$ is completely positive and unital, applying it to each effect of a parent measurement yields a parent measurement for the transformed family. Compatibility is therefore always preserved. If every finite family of incompatible measurements also remains incompatible under this transformation, we say that $\N$ universally preserves measurement incompatibility.

For every channel, state--channel duality gives a fixed bipartite state such that the normalized conditional states generated by arbitrary measurements $M_{a|j}$ on the steering party satisfy $E_{a|j}^{\mathsf T}=\Lambda(M_{a|j})$~\cite{Kiukas2017}. Since transposition preserves joint measurability, these measurements demonstrate steering on this state exactly when $\{\Lambda(M_{a|j})\}$ is incompatible. Universal preservation therefore means that every finite incompatible family demonstrates steering on the same state.
\begin{theorem}
\label{thm:main}
The channel $\N$ universally preserves measurement incompatibility if and only if there exists an isometry $W:\C^d\to\C^D$ such that
\begin{equation}
 \N(W\rho W^\dagger)=\rho \quad\text{for every state }\rho\text{ on }\C^d.
 \label{eq:code}
\end{equation}
\end{theorem}
Here $W$ encodes a quantum state into the input space, and its range is the code subspace. The same encoding works for every $\rho$, and the channel directly outputs the original state.

\emph{Proof of the main theorem.---}
We first prove sufficiency. Equation~\eqref{eq:code} gives $W^\dagger\Lambda(X)W=X$ for every output operator $X$. If the image measurements admit a parent, compressing each of its effects $J$ to $W^\dagger JW$ gives a parent for the original measurements.

Now suppose that $\N$ universally preserves measurement incompatibility. The case $d=1$ is immediate, so assume $d\ge2$. At the compatibility boundary of the binary measurements we construct, one splitting coefficient of their parent vanishes.

We first construct a rank-one minimal informationally complete measurement (MIC) with $d^2$ linearly independent effects $G_i$. Its Hermitian dual element $F_0$ picks out the coefficient of $G_0$ in this basis, since $\Tr(F_0G_i)=\delta_{0i}$. We require $F_0$ to be invertible.

Fix a unit vector $\ket x$, and set $P_0=\proj x$ and $Q=2P_0-\id_d$. For every unit vector $\ket u$ orthogonal to $\ket x$, take the rank-one positive operator
\begin{equation}
 P_u=(\ket x+\ket u)(\bra x+\bra u).
 \label{eq:vectors}
\end{equation}
These operators span the orthogonal complement of $Q$ with respect to the trace inner product. Indeed, if a Hermitian operator $A$ satisfies $\Tr(AP_u)=0$ for every such $\ket u$, varying the phase of $\ket u$ gives $\bra xA\ket u=0$. The remaining condition is $\bra uA\ket u=-\bra xA\ket x$, so $A=\bra xA\ket x\,Q$.

Select $d^2-1$ linearly independent operators from this family and denote them by $P_1,\ldots,P_{d^2-1}$. Since $\Tr(QP_0)=1$, adjoining $P_0$ gives a basis of all Hermitian operators. Thus $T=\sum_iP_i>0$, and we may define
\begin{equation}
 G_i=T^{-1/2}P_iT^{-1/2},\qquad
 F_0=T^{1/2}QT^{1/2}.
 \label{eq:mic}
\end{equation}
The effects $G_i$ have rank one, are linearly independent, and sum to $\id_d$, while $\Tr(F_0G_i)=\Tr(QP_i)=\delta_{0i}$. The invertibility of $F_0$ follows directly from that of $T$ and $Q$.

Moreover, $G_0$ cannot be a projection: otherwise the remaining effects would be supported on its orthogonal complement, so all the $G_i$ would have zero matrix elements between the two subspaces. This contradicts the fact that they span all Hermitian operators. Hence $s_{\max}:=\lambda_{\max}(G_0)^{-1}>1$. Fix the binary measurements $\{G_i,\id_d-G_i\}$ for $i\ge1$, and add $\{B_s,\id_d-B_s\}$, where
\begin{equation}
B_s=\frac{\id_d-sG_0}{2},\qquad 0\le s\le s_{\max}.
\label{eq:family}
\end{equation}

Consider a joint parent whose outcomes are labelled by outcome strings for the fixed measurements. A joint effect corresponding to two first outcomes is bounded above by two rank-one effects with different supports and must therefore vanish. The marginal relations and normalization then fix the remaining joint effects as $G_i$ and $G_0$, respectively. Adding $B_s$ consequently amounts to splitting each $G_i$ into two positive parts. Since $G_i$ has rank one, each part must be a scalar multiple of $G_i$~\cite{Pellonpaa2014}. Thus the splitting is possible if and only if all coefficients of $B_s$ in the MIC basis belong to $[0,1]$. These coefficients are $(1-s)/2$ and $1/2$, so the full family is compatible if and only if $s\le1$.

Take a Stinespring isometry $V:\C^D\to\C^d\otimes\C^r$ for the channel~\cite{Stinespring1955}, so that
\[
 \Lambda(X)=V^\dagger(X\otimes\id_r)V.
\]
Before the channel, a positive part of $G_i$ was specified by a single number. Now any environment effect $0\le H_i\le\id_r$ gives a positive part $V^\dagger(G_i\otimes H_i)V$ of $\Lambda(G_i)$. At $s=1$ only the zeroth scalar coefficient has reached zero.

The projector onto the complement of $\ran V$ compresses to zero. Its expansion lets us change the splitting while keeping the image effect fixed. Since $F_0$ is invertible, we write
\begin{equation}
 \id-VV^\dagger=\sum_iF_0G_i\otimes A_i.
 \label{eq:complement-expansion}
\end{equation}
Left multiplication by $F_0^{-1}\otimes\id_r$, followed by compression with $V$, gives zero because $(\id-VV^\dagger)V=0$. Taking the system partial trace of Eq.~\eqref{eq:complement-expansion} itself gives \mbox{$A_0=\Tr_{\rm sys}(\id-VV^\dagger)$}, since $\Tr(F_0G_i)=\delta_{0i}$. Taking Hermitian parts and setting $C_i=(A_i+A_i^\dagger)/2$, we obtain
\begin{equation}
 \sum_iV^\dagger(G_i\otimes C_i)V=0,\qquad C_0=\Tr_{\rm sys}(\id-VV^\dagger)\ge0.
 \label{eq:coefficient}
\end{equation}

At $s=1$, the splitting coefficients of $\Lambda(B_1)$ are $0$ and $\id_r/2$. Equation~\eqref{eq:coefficient} allows us to add $tC_i$ to the respective coefficients without changing the total effect. If $C_0>0$, choosing $t>0$ sufficiently small makes every coefficient lie strictly between $0$ and $\id_r$. Subtracting $(s-1)\id_r/2$ from the zeroth coefficient then gives a splitting of $\Lambda(B_s)$. For $s>1$ sufficiently close to $1$, all these coefficients remain valid effects. The image family is therefore compatible while the original family is incompatible, contradicting the assumption.

Thus $C_0$ has a unit vector $\ket e$ in its kernel. For any environment vector $\ket u$ and any orthonormal output basis $\{\ket k\}$, $\bra uC_0\ket u=\sum_k\norm{(\id-VV^\dagger)(\ket k\otimes\ket u)}^2$. At $\ket u=\ket e$ every term vanishes, so $\ket\psi\otimes\ket e$ belongs to $\ran V$ for every $\ket\psi$. Define $W\ket\psi=V^\dagger(\ket\psi\otimes\ket e)$. Then $VW\ket\psi=\ket\psi\otimes\ket e$, so $W$ is an isometry and $W^\dagger\Lambda(X)W=X$. Trace duality gives Eq.~\eqref{eq:code}. We have obtained the code from the boundary at $s=1$ alone. By sufficiency, a channel that preserves this one boundary preserves every incompatible family.

\emph{Incompatibility in proper subspaces.---}
For a subspace inclusion $V$, $C_0$ equals the codimension. This will let us choose a single setting $s>1$ for all proper subspaces.

\begin{corollary}
\label{cor:uola}
For every $d\ge2$, there exists a family of $d^2$ binary measurements that is incompatible on $\C^d$ and jointly measurable after compression to any proper subspace.
\end{corollary}

Let $\Pi$ be the orthogonal projection onto a proper subspace. Apply the construction in the main proof with $V$ the inclusion of this subspace. The environment is then one-dimensional, and $VV^\dagger=\Pi$. The coefficients in Eq.~\eqref{eq:coefficient} become real numbers, denoted by $C_i(\Pi)$, with $C_0(\Pi)=\Tr(\id_d-\Pi)\ge1$. Using the vanishing sum in the same equation gives
\begin{equation}
 \begin{aligned}
 \Pi B_{1+t}\Pi
 ={}&t\left(C_0(\Pi)-\frac12\right)\Pi G_0\Pi\\
 &+\sum_{i\ge1}\left(\frac12+tC_i(\Pi)\right)\Pi G_i\Pi.
 \end{aligned}
 \label{eq:compressedparent}
\end{equation}

These coefficients arise from expansion in a fixed basis and depend continuously on $\Pi$. The orthogonal projections onto proper subspaces form a compact set, so all $C_i(\Pi)$ have a common finite bound. Together with $C_0(\Pi)\ge1$, this allows us to choose $t>0$ sufficiently small, independently of $\Pi$, so that every coefficient in Eq.~\eqref{eq:compressedparent} belongs to $[0,1]$ and $1+t<s_{\max}$.

Since $\sum_i\Pi G_i\Pi=\Pi$, the compressed effects form a measurement on the subspace. They realize each fixed measurement, and Eq.~\eqref{eq:compressedparent} shows that they also realize the added measurement by classical postprocessing. Thus the full compressed family is compatible. The original family is incompatible because $1+t>1$, and the choice of $t$ is independent of the subspace. This proves the corollary.

\emph{Encoding and error bounds.---}We now use the same environment operator to bound the transmission error when the compatibility threshold moves.

Let $\tau_\Lambda$ be the largest $s\in[0,s_{\max}]$ for which the image family is jointly measurable, allowing arbitrary parent measurements, and set $\delta=\tau_\Lambda-1$. Retain $C_0$ from Eq.~\eqref{eq:coefficient} and write $\varepsilon=\lambda_{\min}(C_0)$. Equivalently, $\varepsilon$ is the least total squared distance of the vectors $\ket k\otimes\ket e$ from $\ran V$ as the unit vector $\ket e$ varies. The diamond norm gives the worst-case trace-norm error, including inputs entangled with a reference.

\begin{corollary}
\label{cor:stability}
There is a constant $c>0$, depending only on the chosen MIC, such that $\varepsilon\le c^2\delta^2$ whenever $\tau_\Lambda<s_{\max}$. If, in addition, $c^2\delta^2<1$, there is an isometry $W$ such that, with $\mathcal W(\rho)=W\rho W^\dagger$,
\begin{equation}
 \norm{\N\circ\mathcal W-\operatorname{id}_d}_\diamond\le2\varepsilon\le2c^2\delta^2,
 \label{eq:stability}
\end{equation}
where $\operatorname{id}_d$ is the identity channel.
\end{corollary}

We first bound $\varepsilon$. Before the channel the fixed measurements determine the parent, and the coefficient of $G_0$ in $B_1$ is zero. Increasing $s$ from $1$ to $1+t$ subtracts $tG_0/2$ from this effect. After the channel, the same change can be distributed among the other outcomes of the parent by varying their environment effects.

The case $\varepsilon=0$ gives the exact code. For $\varepsilon>0$ we normalize the perturbation in Eq.~\eqref{eq:complement-expansion} so that its zeroth coefficient is $\widehat C_0=\id_r$. The Hermitian coefficients still satisfy the vanishing sum in Eq.~\eqref{eq:coefficient}. The \hyperref[app:estimates]{End Matter} shows that $\|\widehat C_i\|\le c/\sqrt\varepsilon$ for $i\ge1$, with $c$ fixed by the MIC. The square root enters because the complement projector measures the squared norm of the component outside $\ran V$.

Thus $(\id_r+t\widehat C_i)/2$ is an environment effect whenever $0\le t\le\sqrt\varepsilon/c$. Use these effects for $i\ge1$ in the parent splitting, and set the zeroth coefficient to zero. The resulting parts sum to $\Lambda(B_{1+t})$. This parent already realizes the entire interval up to $1+t$, provided $1+t\le s_{\max}$. Allowing all parents therefore gives
\begin{equation}
 \tau_\Lambda\ge\min\left\{s_{\max},\,1+\frac{\sqrt\varepsilon}{c}\right\}.
 \label{eq:threshold-bound}
\end{equation}
The assumption $\tau_\Lambda<s_{\max}$ then gives $\varepsilon\le c^2\delta^2$.

On the shared state, a sufficiently small shift of the steering threshold thus certifies a code with bounded transmission error. To identify its input subspace from the dilation, take a unit eigenvector $\ket e$ of $C_0$ with eigenvalue $\varepsilon$, and set $L=V^\dagger(\id_d\otimes\ket e)$. The operator $S=L^\dagger L$ satisfies $0\le S\le\id_d$ because $VV^\dagger$ is a projection. Moreover, the definition of $C_0$ gives $\Tr(\id_d-S)=\varepsilon$, so $S\ge(1-\varepsilon)\id_d$. When $\varepsilon<1$,
\begin{equation}
 W=LS^{-1/2}
 \label{eq:approxencoder}
\end{equation}
is an isometry.

The component of $VW$ along the environment vector $\ket e$ is $L^\dagger W=\sqrt S$. Hence the encoded channel has the form $(\N\circ\mathcal W)(\rho)=\sqrt S\rho\sqrt S+\mathcal R(\rho)$, where $\mathcal R$ is completely positive and $\mathcal R^*(\id_d)=\id_d-S$. For every logical input, the probability of finding the environment orthogonal to $\ket e$ is therefore at most $\varepsilon$. In the common branch each unit state vector changes by at most $1-\sqrt{1-\varepsilon}$ in norm before normalization. The \hyperref[app:estimates]{End Matter} bounds the combined trace-norm error by $2\varepsilon$, including inputs entangled with a reference, proving Eq.~\eqref{eq:stability}.

\emph{Discussion and outlook.---}
By varying one member of a finite family of binary measurements, we locate the setting at which the measurements cease to admit a common parent. The channel can erase incompatibility by enlarging the range of settings a common parent can reproduce, so a boundary that stays fixed constrains what the environment can learn. We show that preserving it forces the input to contain a noiseless copy of the entire output space, with every superposition transmitted faithfully and the environment left in the same state for every encoded input. If the input and output dimensions agree, the code fills the input and the channel must be unitary. This leaves a different possibility in a larger input space, where the channel can erase information outside the code even though every incompatible measurement family survives. The same construction also gives a single family in every dimension $d\ge2$ that is incompatible on the full space but compatible after compression to any proper subspace.

The construction also gives an explicit encoding. Given a Stinespring isometry $V$ for the channel, choose a unit environment vector $\ket e$ in the kernel of $C_0=\Tr_{\rm sys}(\id-VV^\dagger)$. The same environment state accompanies every logical input under the encoding
\[
W\ket\psi=V^\dagger(\ket\psi\otimes\ket e).
\]

The proof uses $d^2$ binary measurements, one of which varies with a parameter. This number arises because the parent measurement spans the space of Hermitian operators. Whether fewer measurements suffice for this criterion remains an interesting question. This threshold gives a quadratic upper bound on transmission error for small shifts, even for inputs entangled with a reference. The optimal prefactor remains open.

These results give the compatibility threshold a second physical interpretation through the correspondence between measurement incompatibility and steering~\cite{Kiukas2017}. In particular, if the channel admits the required noiseless encoding, every incompatible measurement family demonstrates steering with the corresponding shared state. Otherwise, some parameter values in our construction yield incompatible measurements that cannot demonstrate steering with that state. Thus, examining the steering threshold of this finite measurement family on that state suffices to determine whether the required encoding exists. Measurement incompatibility and steering thus reveal when an entire quantum system can pass through a noisy channel intact.

\begin{acknowledgments}
This work is supported by the Quantum Science and Technology--National Science and Technology Major Project (Grant No. 2024ZD0301000), and the National Natural Science Foundation of China (Grant No. 12275136).
\end{acknowledgments}

\renewcommand{\bibsection}{%
  \par\medskip
  \noindent\makebox[\columnwidth]{\rule{0.28\columnwidth}{0.4pt}}\par
  \phantomsection\addcontentsline{toc}{section}{\refname}%
  \medskip
}

\onecolumngrid
\phantomsection\label{app:estimates}
\begin{center}
\small\bfseries End Matter\par
\end{center}
\twocolumngrid
\appendix
\stepcounter{section}

\emph{Range of the parent splitting.---}Let $K=\id-VV^\dagger$, assume $\varepsilon>0$, and set $Z=(F_0^{-1}\otimes C_0^{-1})K$. Since $KV=0$, the Hermitian part of $Z$ compresses to zero under $V$. Expanding it with Eq.~\eqref{eq:complement-expansion} gives $\operatorname{Re} Z=\sum_iG_i\otimes\widehat C_i$, where
\begin{equation}
 \widehat C_i=\tfrac12(C_0^{-1}A_i+A_i^\dagger C_0^{-1}),\qquad \widehat C_0=\id_r .
 \label{eq:normalized-coefficients}
\end{equation}
Here $A_0=C_0$ by the system partial trace in Eq.~\eqref{eq:coefficient}. The compressed sum vanishes even though its zeroth coefficient is now the identity. In particular,
\begin{equation}
 \sum_{i\ge1} V^\dagger(G_i\otimes\widehat C_i)V=-\Lambda(G_0).
 \label{eq:redistribution}
\end{equation}
Adding $t/2$ times this identity to the splitting of $\Lambda(B_1)$ gives precisely $\Lambda(B_{1+t})$. The only remaining condition on $t$ within the parameter range is positivity of the two parts of each parent outcome.

To bound these coefficients, take an orthonormal output basis $\{\ket j\}$ and write $K_{jj}=(\bra j\otimes\id_r)K(\ket j\otimes\id_r)$. These blocks are positive and sum to $C_0$. For $u=\sum_j\ket j\otimes u_j$, the projection identity $K^2=K$ and Cauchy--Schwarz give
\begin{equation}
 \begin{aligned}
 \langle u|K|u\rangle
 &=\Big\|\sum_j K(\ket j\otimes u_j)\Big\|^2\\
 &\le d\sum_j\langle u_j|K_{jj}|u_j\rangle
 \le d\langle u|(\id_d\otimes C_0)|u\rangle .
 \end{aligned}
 \label{eq:block-estimate}
\end{equation}
Multiplying the block bound on both sides by $F_0^{-1}\otimes C_0^{-1}$ and using $K^2=K$ gives
\begin{equation}
 ZZ^\dagger\le dF_0^{-2}\otimes C_0^{-1},\qquad
 \|Z\|^2\le\frac{d}{\varepsilon\,\lambda_{\min}(F_0^2)}.
 \label{eq:normalized-norm}
\end{equation}
For the Hermitian dual basis $\Tr(F_iG_j)=\delta_{ij}$, we recover $\widehat C_i$ by taking the system partial trace of $(F_i\otimes\id_r)\operatorname{Re}Z$. In an eigenbasis of $F_i$, its expectation in any unit environment vector is a sum of diagonal matrix elements of $\operatorname{Re}Z$ weighted by the eigenvalues of $F_i$.

Each matrix element has magnitude at most $\|Z\|$. With $\|F_i\|_1=\Tr\sqrt{F_i^2}$ this gives
\begin{equation}
 \|\widehat C_i\|\le \sqrt{d/\lambda_{\min}(F_0^2)}\,\max_{i}\|F_i\|_1/\sqrt{\varepsilon}\qquad(i\ge1).
 \label{eq:coefficient-bound}
\end{equation}

\emph{Error of the encoded channel.---}Complete $\ket e$ to an orthonormal environment basis $\{\ket{e_\alpha}\}$ with $\ket{e_0}=\ket e$. The operators $M_\alpha=(\id_d\otimes\bra{e_\alpha})VW$ satisfy $M_0=\sqrt S$ and $\sum_\alpha M_\alpha^\dagger M_\alpha=\id_d$. After tracing out the environment, the other outcomes contribute
\begin{equation}
 \mathcal R(X)=\sum_{\alpha>0}M_\alpha X M_\alpha^\dagger,\qquad
 \mathcal R^*(\id_d)=\id_d-S .
 \label{eq:remainder}
\end{equation}
To include the distortion within the common outcome, set $a=\sqrt{1-\varepsilon}$ and $T=\sqrt S-a\id_d$, so $T\ge0$. The map $\mathcal P(X)=\mathcal R(X)+TXT$ is completely positive, with
\begin{equation}
 H=\mathcal P^*(\id_d)=\id_d-S+T^2=(1+a^2)\id_d-2a\sqrt S .
 \label{eq:square-completion}
\end{equation}
Here $H\ge0$ and $\sqrt S\ge a\id_d$ gives $H\le\varepsilon\id_d$. Expanding $TXT$ now gives, for every operator $X$,
\begin{equation}
 (\N\circ\mathcal W)(X)-X=\mathcal P(X)-\tfrac12(HX+XH).
 \label{eq:error-decomposition}
\end{equation}
To keep the reference explicit, write $\mathcal P(X)=\operatorname{Tr}_E(UXU^\dagger)$, with $U^\dagger U=H$. For any joint operator $Y$ of the logical input and a reference $A$, contraction of the trace norm under partial trace gives
\begin{equation}
 \|(\mathcal P\otimes\operatorname{id}_A)(Y)\|_1
 \le \|U\|^2\|Y\|_1=\|H\|\,\|Y\|_1 .
 \label{eq:reference-estimate}
\end{equation}
Left and right multiplication by $H\otimes\id_A$ each obey the same bound. In Eq.~\eqref{eq:error-decomposition} the positive map thus contributes at most $\varepsilon\|Y\|_1$ and the two remaining terms together contribute at most another $\varepsilon\|Y\|_1$. Taking the supremum over $A$ and over $\|Y\|_1=1$ gives Eq.~\eqref{eq:stability}.

\end{document}